\documentclass[11pt]{amsart}
\usepackage[margin=1in]{geometry}
\usepackage{amsfonts}
\usepackage{amssymb}
\usepackage[dvips]{graphics}
\usepackage{epsfig}
\usepackage{euscript}
\usepackage{color}

\newcommand{\A}{\mathcal{A}}

\newtheorem{theorem}{Theorem}[section]

\newtheorem{lemma}[theorem]{Lemma}

\newtheorem{proposition}[theorem]{Proposition}
\newtheorem{corollary}[theorem]{Corollary}

\newtheorem*{lem}{Lemma \ref{lem:bd2}}
\newcommand{\rem}[1]{{\bf Remark:}}

\newcommand{\Section}[1]{\setcounter{equation}{0}\section{#1}}

\def\idty{{\mathchoice {\rm 1\mskip-4mu l} {\rm 1\mskip-4mu l} %
{\rm 1\mskip-4.5mu l} {\rm 1\mskip-5mu l}}}
\newcommand{\be}{\begin{equation}}
\newcommand{\ee}{\end{equation}}
\newcommand{\bea}{\begin{eqnarray}}
\newcommand{\eea}{\end{eqnarray}}
\newcommand{\beann}{\begin{eqnarray*}}
\newcommand{\eeann}{\end{eqnarray*}}
\newcommand{\unity}{{1\hskip -3pt \rm{I}}}

\newcommand{\X}{\mathcal{X}}

\renewcommand{\P}{\mathcal{P}}

\begin{document}

\renewcommand{\thefootnote}{\fnsymbol{footnote}}
\title{Approximating the ground state of gapped quantum spin systems}

\author{Eman Hamza}
\address{ Department of Mathematics\\
Michigan State University\\
East Lansing, MI 48823, USA} \email{eman@math.msu.edu}
\author{Spyridon Michalakis}
\address{Los Alamos National Labs \\
Los Alamos, NM 87545, USA}
\email{spiros@lanl.gov}
\author{Bruno Nachtergaele}
\address{Department of Mathematics\\
University of California Davis\\
Davis, CA 95616, USA} \email{bxn@math.ucdavis.edu}
\author{Robert Sims}
\address{Department of Mathematics\\
University of Arizona\\
Tucson, AZ 85721, USA}
 \email{rsims@math.arizona.edu}
\date{Version: \today }
\maketitle
\bigskip
\begin{abstract}
We consider quantum spin systems defined on finite
sets $V$ equipped with a metric. In typical examples, $V$ is a
large, but finite subset of $\mathbb{Z}^d$. For finite range Hamiltonians with
uniformly bounded interaction terms and a unique, gapped ground
state, we demonstrate a locality property of the corresponding
ground state projector. In such systems, this ground state projector
can be approximated by the product of observables with quantifiable
supports. In fact, given any subset $\mathcal{X} \subset
V$ the ground state projector can be approximated by the
product of two projections, one supported on $\mathcal{X}$ and one supported on
$\mathcal{X}^c$, and a bounded observable supported on a boundary region in such a
way that as the boundary region increases, the approximation becomes
better. Such an approximation was useful in proving an area law in
one dimension, and this result corresponds to a multi-dimensional
analogue.
\end{abstract}

\maketitle

\footnotetext[1]{Copyright \copyright\ 2009 by the authors. This
paper may be reproduced, in its entirety, for non-commercial
purposes.}

%
%
%
%

\Section{Introduction}\label{sec:intro}

The intense interest in entangled states for purposes of quantum
information and computation during the past decade has stimulated
new investigations in the structure of ground states of quantum spin
systems. It was soon found that the focus on entanglement, and the
mathematical structures related to this concept, also provide a new
and useful way to investigate important physical properties relevant
for condensed matter physics \cite{anders2006,murg2005}. In particular,
progress was made in our understanding of constructive
approximations of the ground states of quantum spin models, and
about the computational efficiency of numerical algorithms to
compute ground state properties and to simulate the time evolution
of such systems \cite{verstraete2004a,verstraete2004b}.

It is now understood that there is a relationship between the amount
of entanglement in a state and the degree of difficulty or the
amount of resources needed to construct a good approximation of it
\cite{schuch2007}. To make this more precise one has to use a quantitative
measure of entanglement. For pure states the {\em entanglement
entropy} is a very natural such measure for bi-partite entanglement.
For a two-component system with components $A$ and $B$ (e.g., $A$
and $B$ label a partition of the spin variables into two sets), the
entanglement entropy of the system in a pure state $\psi$ with
respect to this partition is defined as the von Neumann entropy of
the density matrix describing the restriction of $\psi$ to subsystem
$A$. Mathematically, this density matrix, $\rho_A$, is obtained as
the partial trace of the orthogonal projection onto the state $\psi$
calculated by tracing out the degrees of freedom in $B$.

It was conjectured that under rather general assumptions, the
entanglement entropy associated with $A$ satisfies an {\em Area
Law}. By Area Law one means that if $A$ is associated with a
subvolume of a fixed system in a ground state then, with a suitable
definition of boundary, the entanglement entropy of $A$ should be
bounded by a constant times the size of the boundary of $A$ (e.g.,
the surface area in the case of a three-dimensional volume).

One way to understand this conjecture is to observe that it holds
almost trivially for the so-called Matrix Product States (MPS), aka
Finitely Correlated States \cite{fannes1992} and their generalizations in
higher dimensions known as PEPS \cite{perez2008}. It is therefore
plausible that to prove the Area Law it is key to have a good handle
on MPS-like approximations of a general ground state. By using
such a strategy Hastings obtained a proof of the Area Law for
one-dimensional systems \cite{Hastarea}.

The goal of this paper is to prove a generalization to arbitrary
dimensions of the approximation result Hastings used in his proof of
the Area Law in one dimension. This is Theorem~\ref{thm:gsapp}
below, which demonstrates that the ground state projector of certain
gapped quantum spin systems can be approximated by a product of
three projections with known supports. We describe this result in
detail in the next section.

%
%
%
%

\Section{A Ground State Approximation Theorem} \label{sec:gsat}

We begin with a brief description of the quantum spin systems that
will be considered in this work. Let $\mathcal{V}$ be a countable, locally finite
set equipped with a metric $d$. In most examples, $\mathcal{V} = \mathbb{Z}^{\nu}$ and $d$ is, for example, the
Euclidean metric. Since our arguments do not make any use of the structure
of the underlying lattice, we present our models in this more general setting.
Most of our analysis applies on finite subsets $V \subset \mathcal{V}$.
More specifically, to each $x\in V$, we will associate a finite-dimensional
Hilbert space $\mathcal{H}_x$; the dimension of $\mathcal{H}_x$ will
be denoted by $n_x$. Set $M_{n_x}$ to be the complex $n_x \times
n_x$ matrices defined over $\mathcal{H}_x$. We denote by
$\mathcal{H}_V = \bigotimes_{x \in V} \mathcal{H}_x$ the
Hilbert space of states over $V$, and similarly $\A_V = \bigotimes_{x \in
V}M_{n_x}$ is the algebra of local observables. For any finite subsets 
$X \subset Y \subset \mathcal{V}$, the local algebra $\A_X$ can be embedded in 
$\A_Y$ by using that any $A \in \A_X$ corresponds to $A \otimes \unity \in \A_Y$.
Thus, the algebra of quasi-local observables, $\A_{\mathcal{V}}$, can be 
defined as the norm-closure of the union of all local algebras; the union 
taken over all finite subsets of $\mathcal{V}$. Moreover, we say that
an observable $A \in \mathcal{A}_{\mathcal{V}}$ is supported in $X \subset \mathcal{V}$,
denoted by $\mbox{supp}(A) =X$, if $X$ is the minimal set for which $A$ can be written as $A = \tilde{A} \otimes \unity$ with
$\tilde{A} \in \A_X$.

A quantum spin model is then defined by its interaction and the
corresponding local Hamiltonians. An interaction is function $\Phi$
from the set of finite subsets of $\mathcal{V}$ into $\A_{\mathcal{V}}$ with the property that
given any finite $V \subset \mathcal{V}$, $\Phi(V) = \Phi(V)^* \in \mathcal{A}_V$.
Given such an interaction, one can associate a family of local Hamiltonians, parametrized by the finite
subsets of $\mathcal{V}$, defined by setting
\begin{equation} \label{eq:locham}
H_V = \sum_{X \subset V} \Phi(X) 
\end{equation}
for finite $V \subset \mathcal{V}$.
Since each local Hamiltonian $H_V$ is self-adjoint, it generates a one 
parameter group of automorphisms, which we will denote by $\tau^V_t$,
that is often called the finite volume dynamics. This dynamics is defined by setting 
\begin{equation}
\tau_t^V(A) = e^{itH_V} A e^{-itH_V}, \quad A \in \mathcal{A}_V.
\end{equation}
When the subset $V$ on which the dynamics is defined has been fixed, we
will often write $\tau_t$ to ease the notation.

Theorem~\ref{thm:gsapp} below holds for certain quantum spin models.
Our first assumption pertains to the underlying set $\mathcal{V}$. We must assume
that there exists a number $\mu_0 >0$ for which
\begin{equation} \label{def:kappa}
\kappa_{\mu_0} = \sup_{x \in \mathcal{V}} \sum_{y \in \mathcal{V}} e^{- \mu_0 d(x,y)} < \infty \, .
\end{equation}
Next, we make the assumptions on the interactions we consider precise.
\newline {\bf A1:} We consider interactions that are of finite range.
Specifically, we assume that there exists a number $R>0$, called the
range of the interactions, for which $\Phi(X) = 0$ if the diameter
of $X$ exceeds $R$. Here, for any finite $X \subset \mathcal{V}$,  the diameter of
$X$ is
\begin{equation}
\mbox{diam} (X) = \max \{ d(x,y) \, : \, x,y\in X\} \,
\end{equation}
where $d$ is the metric on $\mathcal{V}$. 

\noindent {\bf A2:} We will only consider uniformly bounded
interactions. Hence, there is a number $J>0$ such that for any finite $X
\subset \mathcal{V}$,  $\| \Phi(X) \| \leq J$.

\noindent {\bf A3:} Let $\chi_{\Phi}$ be the characteristic function defined
over the set of finite subsets of $\mathcal{V}$, i.e. for any finite
$X \subset \mathcal{V}$,
$\chi_{\Phi}(X) = 1$ if $\Phi(X) \neq 0$, and $\chi_{\Phi}(X) = 0$
otherwise. We assume the following quantity is finite
\begin{equation}\label{eq:finiteint2}
N_{\Phi} = \sup_{x \in \mathcal{V}} \sum_{\stackrel{X \subset \mathcal{V}:}{x \in X}}
|X| \chi_{\Phi}(X) < \infty 
\end{equation}
where the sum above is taken over finite subsets $X \subset \mathcal{V}$ and $|X|$ is the 
cardinality of $X$. Another relevant quantity, which often appears in our estimates, is
\begin{equation}\label{eq:finiteint1}
C_{\Phi} = \sup_{x \in \mathcal{V}} \sum_{\stackrel{X \subset \mathcal{V}:}{x \in X}}
\chi_{\Phi}(X) \leq N_{\Phi} \, .
\end{equation}

\noindent {\bf A4:} We will assume that a given local Hamiltonian $H_V$ has a unique, normalized
ground state which we denote by $\psi_0$. $P_0$ will be the
projection onto this ground state in $\mathcal{H}_V$, and we will label by $\gamma$ the
length of the gap to the first excited state.

Much progress has been made recently in proving locality estimates
for general quantum spin systems, see e.g. \cite{Hastkoma,nasi1}. The following estimate
was proven in \cite{nasi4}.

\begin{theorem}[Lieb-Robinson Bound] \label{thm:lrb}
Let $\mathcal{V}$ be a countable, locally finite set equipped with a metric $d$ for which
(\ref{def:kappa}) holds. Let $\Phi$ be an interaction on $\mathcal{V}$ that satisfies assumptions {\bf
A1}, {\bf A2}, and {\bf A3}. For every $\mu \geq \mu_0$, there exists numbers $c$ and $v$ 
such that given any finite $V \subset \mathcal{V}$ and any observables $A \in \A_X$ and
$B \in \A_Y$ with $X, Y \subset V$ and $X \cap Y = \emptyset$, the bound
\begin{equation}
\left\|  \left[ \tau_t^V(A), B \right] \right\| \leq c \, \| A \| \,
\| B \| \, \min \left[ | \partial_{\Phi} X|, | \partial_{\Phi} Y| \right] e^{- \mu
\left( d(X,Y) - v |t| \right)}
\end{equation}
holds for all $t \in \mathbb{R}$. 
\end{theorem}
Here for any finite $X \subset \mathcal{V}$, the set 
\begin{equation}
\partial_{\Phi} X = \{ x \in X : \mbox{ there exists } Y \subset \mathcal{V} \mbox{ with } x \in Y, Y \cap \mathcal{V} \setminus X \neq \emptyset, \mbox{ and}, \Phi(Y) \neq 0 \} \, ,
\end{equation}
the number $d(X,Y) = \min \{ d(x,y) : x \in X, y \in Y \}$, and we stress that 
both $c$ and $v$ are independent of $V \subset \mathcal{V}$.

Before we state Theorem~\ref{thm:gsapp}, some further notation is
necessary. Let $\Phi$ be an interaction which satisfies assumptions
{\bf A1} - {\bf A4} above. Again, most of our analysis applies on finite 
sets $V \subset \mathcal{V}$, and often the particular finite set under
consideration is the one whose existence is guaranteed by {\bf A4}.
In general, for finite sets $X \subset V \subset \mathcal{V}$,
we will denote by $\partial X$ the R-boundary of $X$ in $V$, i.e.,
\begin{equation}
\partial X = \{ x \in X : \mbox{there exists } y \in V \setminus X,  \mbox{with } d(x,y) \leq R\}.
\end{equation}
We note that with $X$ fixed, the set  $\partial X$ is independent of $V$, for $V$ sufficiently large, and so we
suppress this in our notation. To state our main result, we fix an additional subset $\X \subset V$.
The following sets, defined relative to this fixed set $\X$ and dependent on a length-scale $\ell >R$, will
also play an important role. Set
\begin{equation} \label{eq:xint}
\X_{\rm int} = \X_{\rm int}^V( \ell) = \left\{ x \in \X :  \mbox{for
all } y \in \partial \X, d(x,y) \geq \ell \right\},
\end{equation}
\begin{equation} \label{eq:xbd}
\X_{\rm bd} = \X_{\rm bd}^V( \ell) = \left\{ x \in V : \mbox{there
exists } y \in
\partial \X, d(x,y) < \ell \right\},
\end{equation}
and
\begin{equation} \label{eq:xext}
\X_{\rm ext} = \X_{\rm ext}^V( \ell) = \left\{ x \in V \setminus \X :
\mbox{for all }y\in
\partial \X, d(x,y) \geq \ell \right\}.
\end{equation}

The sets $\X_{\rm int}$, $\X_{\rm bd}$, and $\X_{\rm ext}$
correspond respectively to the $\ell$-interior, the $\ell$-border, and the
$\ell$-exterior of $\X$ in $V$. We note that for any finite $V \subset \mathcal{V}$ and 
each $\X \subset V$,  the sets $\X_{\rm int}$, $\X_{\rm bd}$, and $\X_{\rm ext}$
are disjoint, and moreover, $V= \X_{\rm int} \cup \X_{\rm bd} \cup
\X_{\rm ext}$. We caution that the notation $\partial \X_{\rm ext}$ here refers to the boundary of this
set as a subset of $V$, as indicated above; not as a subset of $\mathcal{V}$. 
It will also be important for us that there exists a
number $C>0$, independent of both $\ell$ and $V$, for which
\begin{equation} \label{eq:setbd}
\max \left\{ | \partial \X_{\rm int}|, | \X_{\rm bd}|, | \partial
\X_{\rm ext}| \right\} \, \leq \, C \ell | \partial \X| \, .
\end{equation}
The above inequality constitutes the main structural assumption on the set $\mathcal{V}$ and
the reference set $\X \subset V$. We may now state our main result.

\begin{theorem} \label{thm:gsapp}
Let $\mathcal{V}$ be a countable, locally finite set equipped with a metric $d$ for which
(\ref{def:kappa}) holds, and moreover, let $\Phi$ be an interaction on $\mathcal{V}$ 
which satisfies assumptions {\bf A1} - {\bf A4}, take
$V \subset \mathcal{V}$ to be the finite set from {\bf A4}, 
and suppose $\X \subset V$ satisfies (\ref{eq:setbd}).
For any $\ell >R$, there exists two projections $P_{\X} \in \A_{\X}$ and 
$P_{\X^c}\in \A_{V \setminus \X}$ and an observable 
$P_{\X_{\rm bd}} \in \A_{\X_{\rm bd}(3 \ell)}$ with $\| P_{\X_{\rm bd}}  \| \leq 1$
such that 
\begin{equation} \label{eq:p0bd}
\| P_{\X_{\rm bd}} P_{\X} P_{\X^c} - P_0 \| \leq K
C_{\Phi}N_{\Phi}|
\partial \X |^2 \ell^{7/2} e^{- \ell/ 2\xi}.
\end{equation}
Here $\xi$ is defined by
\begin{equation}
\frac{2}{ \xi} = \frac{ \mu \gamma^2}{\mu^2 v^2 + \gamma^2} \, ,
\end{equation}
for $\mu \geq 2 \mu_0$.
\end{theorem}

In many situations, it is useful to have a non-negative approximation of the ground state projector.
For this reason, we also state the following corollary.
\begin{corollary} \label{cor:gsapp} 
Under the assumptions of Theorem~\ref{thm:gsapp}, the estimate
\begin{equation}
\| P_{\X^c} P_{\X} P^*_{\X_{\rm bd}} P_{\X_{\rm bd}} P_{\X} P_{\X^c} - P_0 \| \leq K
C_{\Phi}N_{\Phi}|
\partial \X |^2 \ell^{7/2} e^{- \ell/ 2\xi}
\end{equation}
also holds.
\end{corollary}
Since all the observables involved have norm bounded by 1, Corollary~\ref{cor:gsapp} is an immediate
consequence of Theorem~\ref{thm:gsapp}.

With the methods discussed in \cite{nasi4}, Theorem~\ref{thm:gsapp} can 
easily be generalized to other types of interactions. For example, the finite range condition {\bf A1}
can be replaced with an exponentially decaying analogue. Moreover, the uniform bound on 
the local interactions terms, i.e. {\bf A2}, can also be lifted. One could also consider systems
where the underlying sets satisfied a bound of the form (\ref{eq:setbd}) with a larger power of
$\ell$. This would increase the power of $\ell$ in the statement of Theorem~\ref{thm:gsapp}, but not
effect the exponential decay. We give the argument in the context
described above for convenience of presentation. It would be very interesting to see that the
methods developed here can be used to prove an area law, similar to Hastings' result \cite{Hastarea}, in
dimensions greater than 1, but this is beyond the scope of the present work.

%
%
%
%

\Section{Basic Set-Up} \label{sec:setup}

Before beginning the proof of this theorem, we will introduce the
basic objects to be analyzed. Let $V \subset \mathcal{V}$ be the 
finite set described in {\bf A4}, take $\ell >R$, and 
fix $\X \subset V$ satisfying (\ref{eq:setbd}). Note that  $H_V$, as defined in (\ref{eq:locham}), can be
written as a sum of three local Hamiltonians
\begin{equation} \label{eq:hamdecomp}
H_V = H_{\X_{\rm int}}^b + H_{\X_{\rm bd}} + H_{\X_{\rm ext}}^b,
\end{equation}
where $H_{\X_{\rm bd}}$ is as in (\ref{eq:locham}) with $V = \X_{\rm bd}$, 
and for $Z \in \{ \X_{\rm int}, \X_{\rm ext} \}$,
\begin{equation} \label{eq:ham+bd}
H_Z^b = \sum_{\stackrel{X \subset V:}{X \cap Z \neq \emptyset}}
\Phi(X) \, ,
\end{equation}
corresponds to local Hamiltonians that include boundary terms.
Without loss of generality, we may subtract a constant, namely
$\langle \psi_0, H_V \psi_0 \rangle$, and thereby assume that the
ground state energy of each of these Hamiltonians is zero. In this
case, each term in (\ref{eq:hamdecomp}) has zero ground state
expectation. It is not clear, however, that each of these terms,
when applied to the ground state, have a quantifiably small norm. To
achieve small norms, we introduce non-local versions of these
observables.

Generally, given any self-adjoint Hamiltonian $H$ on
$\mathcal{H}_V$, a local observable $A\in \mathcal{A}_V$, and
$\alpha >0$, we may define a new observable
\begin{equation} \label{eq:smear}
(A)_{\alpha} = \sqrt{ \frac{ \alpha}{ \pi}} \int_{- \infty}^{\infty}
\tau_t(A) e^{- \alpha t^2} dt,
\end{equation}
where here $\tau_t(A) = e^{itH}Ae^{-itH}$. It is easy to see that
for every $\alpha >0$, $\| (A)_{\alpha} \| \leq \| A\|$.

In the present context, we will define the $( \cdot )_{\alpha}$
operation with respect to the Hamiltonian $H_V$ and its
corresponding dynamics $\tau_t^V$. Observe that these non-local
observables still sum to the total Hamiltonian, i.e.,
\begin{equation}\label{eq:nonlocH}
H_V = ( H_{\X_{\rm int}}^b)_{\alpha} \, + \, ( H_{\X_{\rm
bd}})_{\alpha} \, + \, ( H_{\X_{\rm ext}}^b)_{\alpha}.
\end{equation}
They also have zero ground state expectation, i.e.,
\begin{equation} \label{eq:0gse}
\langle \psi_0, ( H_{\X_{\rm int}}^b)_{\alpha} \psi_0 \rangle \, =
\, \langle \psi_0, ( H_{\X_{\rm bd}})_{\alpha} \psi_0 \rangle \, =
\, \langle \psi_0, ( H_{\X_{\rm ext}}^b )_{\alpha} \psi_0 \rangle \,
= \, 0.
\end{equation}

We will prove, in Proposition~\ref{prop:bd}, that each of these
smeared-out Hamiltonians, when applied to the ground state, has a
norm bounded by an exponentially decaying quantity. The cost of this
{\it small norm} estimate is a loss of locality, which is expressed
in terms of the support of these Hamiltonians.

Equipped with the Lieb-Robinson bound in Theorem \ref{thm:lrb}, we
begin the next step in the proof of Theorem \ref{thm:gsapp}. Here we
approximate each of $( H_{\X_{\rm int}}^b)_{\alpha}$, $( H_{\X_{\rm
bd}})_{\alpha}$, and $( H_{\X_{\rm ext}}^b)_{\alpha}$ with local
observables $M_{\X}(\alpha), M_{\X_{\rm bd}(2 \ell)}(\alpha)$ and
$M_{\X^c}(\alpha)$ respectively. These observables will be
constructed in such a way that they not only have a quantifiable
support but also a small vector norm, when applied to the ground
state.

To make these approximating Hamiltonians explicit, we introduce
local evolutions. Define three different dynamics, each acting on
$\A_V$, by setting
\begin{equation} \label{eq:locev1}
\begin{array}{ccc}
\vspace{.2cm}
\tau_t^{\X}(A) &=& e^{itH_{\X}} A e^{-itH_{\X}} \: , \\
\vspace{.2cm}
\tau_t^{\X_{\rm bd}(2 \ell)}(A) &=& e^{itH_{\X_{\rm bd}(2 \ell)}} A e^{-itH_{\X_{\rm bd}(2 \ell)}} \: ,\\
\tau_t^{\X^c}(A) &=& e^{itH_{V \setminus \X}} A e^{-itH_{V \setminus
\X}} \: ,
\end{array}
\end{equation}
for any local observable $A \in \A_V$ and any $t \in \mathbb{R}$.
The set $\X_{\rm bd}(2 \ell)$ is as defined in (\ref{eq:xbd}) with
the length scale doubled, and the local Hamiltonians used in the
evolutions above are as in (\ref{eq:locham}). The approximating
Hamiltonians are given by
\begin{equation} \label{eq:locham1}
\begin{array}{ccc}
\vspace{.2cm} M_{\X}(\alpha) & = &  \sqrt{ \frac{ \alpha}{ \pi}}
\int_{-
\infty}^{\infty} \tau_t^{\X} \left( H_{\X_{\rm int}}^b \right) e^{- \alpha t^2} dt \: ,\\
\vspace{.2cm} M_{\X_{\rm bd}(2 \ell)}(\alpha) & = &  \sqrt{ \frac{
\alpha}{ \pi}} \int_{-
\infty}^{\infty} \tau_t^{\X_{\rm bd}(2  \ell)} \left( H_{\X_{\rm bd}} \right) e^{- \alpha t^2} dt \: ,\\
M_{\X^c}(\alpha) & = &  \sqrt{ \frac{ \alpha}{ \pi}} \int_{-
\infty}^{\infty} \tau_t^{\X^c} \left( H_{\X_{\rm ext}}^b \right)
e^{- \alpha t^2} dt \: .
\end{array}
\end{equation}
The observables defined above, i.e. in (\ref{eq:locham1}), have been chosen
such that $\mbox{supp}(M_{\X}(\alpha)) =  \X$,
$\mbox{supp}(M_{\X_{\rm bd}(2 \ell)}(\alpha)) = \X_{\rm bd}(2
\ell)$, and $\mbox{supp}(M_{\X^c}(\alpha)) = V \setminus \X$. 
To ease notation, we will often write
$M_{\X_{\rm bd}}(\alpha) = M_{\X_{\rm bd}(2 \ell)}(\alpha)$ and
suppress the exact size of the support.

The following technical estimate summarizes the results mentioned above.
\begin{lemma}\label{lem:bd2}
Let $\mathcal{V}$ be a countable, locally finite set equipped with a metric $d$ for which
(\ref{def:kappa}) holds, and let $\Phi$ be an interaction on $\mathcal{V}$ 
which satisfies assumptions {\bf A1} - {\bf A4}. For the finite 
set $V \subset \mathcal{V}$ from {\bf A4} and any set $\X \subset V$
satisfying (\ref{eq:setbd}) for all $\ell >R$, the estimate 
\begin{equation} \label{eq:totbd}
\left\| H_V \, - \, \left(M_{\X}( \alpha) + M_{\X_{\rm bd}}(\alpha)
+M_{\X^c}( \alpha) \right) \right\| \, \leq \, K | \partial \X|
\ell^{3/2} e^{- \frac{\ell}{\xi}},
\end{equation}
holds along the parametrization $2 \alpha \epsilon \ell = \mu v^2$ where $\mu \geq 2 \mu_0$. The numbers $\xi$ and $\epsilon$
are defined in terms of the gap $\gamma$ and the quantities $\mu$ and $v$ from the 
Lieb-Robinson estimates as
\begin{equation}
0 < \frac{2}{\xi} = (1- \epsilon) \mu = \frac{ \gamma^2}{ \mu^2 v^2 + \gamma^2}  \mu \, .
\end{equation}
Along the same parametrization, the bound
\begin{equation} \label{eq:Monpsi}
\max \left\{ \left\| M_{\X}( \alpha) \psi_0 \right\|, \left\|
M_{\X_{\rm bd}}( \alpha) \psi_0 \right\|, \left\| M_{\X^c} ( \alpha)
\psi_0 \right\| \right\}  \, \leq \, K | \partial \X| \ell^{3/2}
e^{- \frac{\ell}{\xi}}.
\end{equation}
also holds.
\end{lemma}
These estimates play an important role in the proof of Theorem~\ref{thm:gsapp}.
We prove this lemma in Section~\ref{sec:test}. For now, we use it to prove
Theorem~\ref{thm:gsapp}.

%
%
%
%

 \Section{Proof of Theorem~\ref{thm:gsapp}}
The first step in verifying the bound claimed in (\ref{eq:p0bd}) is
to find an explicit approximation to the ground state projector
$P_0$. With this in mind, define
\begin{equation}
\P_{\alpha} \, = \, \sqrt{ \frac{\alpha}{\pi}} \int_{-
\infty}^{\infty} e^{i H_V t} e^{- \alpha t^2} \, dt,
\end{equation}
for any $\alpha >0$. Clearly, $\P_{\alpha} \in \A_V$ and for any vectors $f,g \in \mathcal{H}_V$, the
spectral theorem implies
\begin{eqnarray}
\langle f, ( \P_{\alpha} - P_0) g \rangle & = &  \sqrt{
\frac{\alpha}{\pi}} \int_{- \infty}^{\infty} e^{- \alpha t^2} \,
\int_0^{\infty} e^{i \lambda t} d \langle f, E_{\lambda} g \rangle
dt \, - \, \langle f, P_0 g \rangle \nonumber \\ & = &
\int_{\gamma}^{\infty} \sqrt{ \frac{\alpha}{\pi}} \int_{-
\infty}^{\infty} e^{i \lambda t} e^{- \alpha t^2} \, dt \, d \langle
f, E_{\lambda} g \rangle
 \nonumber \\ & = & \int_{\gamma}^{\infty}  e^{- \frac{\lambda^2}{4 \alpha}} \, d \langle f, E_{\lambda} g \rangle,
\end{eqnarray}
where we have denoted by $E_{\lambda}$ the spectral projection
corresponding to $H_V$. This readily yields that
\begin{equation} \label{eq:appbd}
\left\| \P_{\alpha} - P_0 \right\| \, \leq \, e^{- \frac{
\gamma^2}{4 \alpha}},
\end{equation}
where $\gamma$ is the gap of the Hamiltonian $H_V$.

Using the operators introduced in \eqref{eq:locham1}, we define an
analogous ground state approximate by setting
\begin{equation}
\tilde{\P}_{\alpha} \, = \,  \sqrt{ \frac{ \alpha}{ \pi}}
\int_{-\infty}^{\infty} e^{i(M_{\X}+M_{\X_{\rm bd}}+M_{\X^c})t} e^{-
\alpha t^2} \, dt.
\end{equation}
Here we have dropped the dependence of $M_{\X}$, $M_{\X_{\rm bd}}$,
and $M_{\X^c}$ on $\alpha$. Clearly,
\begin{equation} \label{eq:appbd2}
\left\| \tilde{\P}_{\alpha} - P_0 \right\| \, \leq  \, \left\|
\tilde{\P}_{\alpha} - \P_{\alpha} \right\| \, + \, \left\|
\P_{\alpha} - P_0 \right\|.
\end{equation}
The final term above we have bounded in (\ref{eq:appbd}). To bound
the first term in (\ref{eq:appbd2}), we introduce the function
\begin{equation}
F_t( \lambda) \, = \, e^{i \lambda (M_{\X}+M_{\X_{\rm
bd}}+M_{\X^c})t} \, e^{i(1- \lambda) H_V t}.
\end{equation}
One easily calculates that
\begin{equation}
F_t'( \lambda) \, = \, - \, i \, t \,  e^{i \lambda
(M_{\X}+M_{\X_{\rm bd}}+M_{\X^c})t} \left\{ H_V \, - \,
(M_{\X}+M_{\X_{\rm bd}}+M_{\X^c}) \right\} e^{i(1- \lambda)H_V t}.
\end{equation}
Using Lemma~\ref{lem:bd2}, we conclude that, if $2 \alpha \epsilon \ell = \mu v^2$, then
\begin{eqnarray} \label{eq:pbd}
\left\| \tilde{\P}_{\alpha} - \P_{\alpha} \right\| & \leq &  \sqrt{ \frac{ \alpha}{ \pi}} \int_{-\infty}^{\infty} \left\| F_t(1) - F_t(0) \right\| \, e^{- \alpha t^2} \, dt \nonumber \\
& \leq & \sqrt{ \frac{ \alpha}{ \pi}} \int_{-\infty}^{\infty} \left\|  H_V \, - \, (M_{\X}+M_{\X_{\rm bd}}+M_{\X^c}) \right\| \, |t| \, e^{- \alpha t^2} \, dt \nonumber \\
& \leq & K | \partial \X| \ell^2 e^{- \frac{\ell}{\xi}} \, .
\end{eqnarray}
In the bound above, and for the rest of this section, the number $K$ which appears in our estimates will depend
on the parameters of the quantum spin models, but not on $\ell$. It may change from line to line, but we do not
indicate this in our notation. 

Inserting both \eqref{eq:appbd} and \eqref{eq:pbd} into \eqref{eq:appbd2} yields
\begin{equation} \label{eq:appbd3}
\left\| \tilde{\P}_{\alpha} - P_0 \right\| \, \leq  \, K  |
\partial \X| \ell^2 e^{- \frac{\ell}{\xi}},
\end{equation}
again along the parametrization $2 \alpha \epsilon \ell = \mu v^2$.

 Set
\begin{equation}\label{eq:delta}
\delta = \ell^{3/2} e^{- \frac{ \ell}{2\xi}}.
\end{equation}
The projections $P_{\X}$ and $P_{\X^c}$, which appear in the
statement of Theorem~\ref{thm:gsapp}, are defined to be the spectral
projection corresponding to the matrix $M_{\X}$, respectively
$M_{\X^c}$, onto those eigenvalues less than $\delta$. By
definition, $P_{\X} \in \mathcal{A}_{\X}$ and $P_{\X^c} \in
\mathcal{A}_{\X^c}$ as claimed. Moreover, using Lemma~\ref{lem:bd2},
it is easy to see that
\begin{equation} \label{eq:oxbd}
\max_{Z \in \{ \X, \X^c \}} \left\| \left( 1 - P_Z \right) \psi_0
\right\| \, \leq \, \max_{Z \in \{ \X, \X^c \}} \frac{1}{\delta}
\left\| M_Z \psi_0 \right\| \, \leq \, K  | \partial \X|  e^{-
\frac{\ell}{ 2 \xi}},
\end{equation}
by construction; hence this choice of $\delta$.

With this bound, we can insert these projections into our previous
estimates. In fact,
\begin{equation}
\left\| \tilde{\P}_{\alpha} P_{\X} P_{\X^c} - P_0 \right\| \, \leq
\,  \left\| \left( \tilde{\P}_{\alpha} \, - \, P_0 \right) P_{\X}
P_{\X^c} \right\| \, + \,
  \left\| P_0 \left( 1 \, - \, P_{\X} P_{\X^c} \right) \right\| \, .
\end{equation}
The first term above is estimated using (\ref{eq:appbd3}) above.
Since
\begin{equation}
(1-P_{\X} P_{\X^c}) \, = \, \frac{1}{2} \left\{
(1-P_{\X})(1+P_{\X^c}) \, + \, (1-P_{\X^c})(1+P_{\X}) \right\},
\end{equation}
it is clear that
\begin{equation}
 \left\| P_0 \left( 1 \, - \, P_{\X} P_{\X^c} \right) \right\| \, \leq \, \| P_0 (1 - P_{\X}) \| \, + \, \| P_0 (1- P_{\X^c}) \| \, \leq \, 2 K | \partial \X|  e^{- \ell/ 2 \xi},
\end{equation}
using (\ref{eq:oxbd}). Therefore, we now have that
\begin{equation}
\left\| \tilde{\P}_{\alpha} P_{\X} P_{\X^c} - P_0 \right\| \, \leq
\,K | \partial \X| \ell^2 \, e^{- \frac{\ell}{2\xi}},
\end{equation}
again along the parametrization $2 \alpha \epsilon \ell = \mu v^2$.

In order to define the final observable $P_{\X_{\rm bd}}$, we write
\begin{equation}
\tilde{\P}_{\alpha} P_{\X} P_{\X^c} \, = \,  \sqrt{ \frac{ \alpha}{
\pi}} \int_{-\infty}^{\infty} e^{i(M_{\X}+M_{\X_{\rm
bd}}+M_{\X^c})t} e^{-i(M_{\X}+M_{\X^c})t} e^{i(M_{\X}+M_{\X^c})t}
P_{\X} P_{\X^c} e^{- \alpha t^2} \, dt.
\end{equation}
Since the supports of $M_{\X}$ and $M_{\X^c}$ are disjoint, it is
clear that
\begin{equation}
\begin{split}
e^{i(M_{\X}+M_{\X^c})t} P_{\X} P_{\X^c} \, - & \,P_{\X} P_{\X^c} \, = \\
 \frac{1}{2} (e^{iM_{\X} t} P_{\X}- P_{\X})(e^{iM_{\X^c}t}P_{\X^c}+P_{\X^c}) \, &+ \,
\frac{1}{2}(e^{iM_{\X^c}t}P_{\X^c}-P_{\X^c})(e^{iM_{\X}t}P_{\X}+P_{\X})
\, .
\end{split}
\end{equation}
Moreover, for $Z \in \{ \X ,\X^c \}$, we have that
\begin{eqnarray}
\left| \langle f, (e^{iM_Z t}-1)P_Z g \rangle \right| & = & \left| \int_0^{\delta} (e^{i \lambda t} - 1) \,  d \langle f, E_{\lambda}^ZP_Z g \rangle \right| \nonumber \\
& \leq & \delta \, |t| \, \| f \| \, \| g \|,
\end{eqnarray}
and therefore,
\begin{equation}
\left\| e^{i(M_{\X}+M_{\X^c})t} P_{\X} P_{\X^c} \, -  \,P_{\X}
P_{\X^c} \right\| \, \leq \, 2 \delta |t| \, .
\end{equation}
If we now define the operator,
\begin{equation} \label{eq:phat}
\hat{\P}_{\alpha} \, = \, \sqrt{ \frac{ \alpha}{ \pi}}
\int_{-\infty}^{\infty} e^{i(M_{\X}+M_{\X_{\rm bd}}+M_{\X^c})t}
e^{-i(M_{\X}+M_{\X^c})t} e^{- \alpha t^2} \, dt,
\end{equation}
then we have just demonstrated that
\begin{eqnarray}
\| \hat{\P}_{\alpha} P_{\X} P_{\X^c} - P_0 \| & \leq & \| \hat{\P}_{\alpha}  P_{\X} P_{\X^c} - \tilde{\P}_{\alpha}P_{\X}P_{\X^c} \| \, + \, \| \tilde{\P}_{\alpha}P_{\X}P_{\X^c} - P_0 \| \nonumber \\
& \leq &  2 \delta \sqrt{ \frac{ \alpha}{ \pi}}
\int_{-\infty}^{\infty}
|t|  e^{- \alpha t^2} \, dt \, + \, K | \partial \X| \ell^2 \, e^{- \frac{\ell}{2\xi}} \nonumber \\
&\leq & K | \partial \X| \ell^2 \, e^{- \frac{\ell}{2\xi}}.
\end{eqnarray}

The proof of Theorem~\ref{thm:gsapp} is complete if we show that the
operator $\hat{\P}_{\alpha}$, defined in (\ref{eq:phat}) above,
is well approximated by a local observable. In fact, below we introduce 
an observable $P_{\X_{\rm bd}} \in \A_{\X_{\rm bd}(3\ell)}$ with $\| P_{\X_{\rm bd}} \| \leq 1$ for which 
\begin{equation} \label{eq:goal}
\left\| P_{\X_{\rm bd}}  - \hat{P}_{\alpha} \right\| \, \leq
\, K | \partial \X|^2 \ell^{7/2} e^{- \frac{\ell}{\xi}} \, .
\end{equation}
{F}rom this estimate, Theorem~\ref{thm:gsapp} easily follows as
\begin{eqnarray}
\|  P_{\X_{\rm bd}}P_{\X} P_{\X^c} - P_0 \| & \leq & \left\| \left(P_{\X_{\rm bd}} - \hat{\P}_{\alpha} \right) P_{\X}P_{\X^c} \right\| \, + \, \| \hat{\P}_{\alpha}P_{\X}P_{\X^c} - P_0 \| \nonumber \\
& \leq & K  | \partial \X|^2 \ell^{7/2} e^{- \frac{\ell}{2
\xi}} \, .
\end{eqnarray}
We need only prove that such an observable $P_{\X_{\rm bd}}$ exists.

The existence of this observable is simple. We just take 
the normalized partial trace of $\hat{P}_{\alpha}$ over 
the complimentary Hilbert space, i.e., the one associated 
with $\X_{\rm bd}(3\ell)^c \subset V$. To prove the estimate in
(\ref{eq:goal}), it is convenient to calculate this partial trace
as an integral over the group of unitaries, technical details on
this may be found in \cite{bravyi2006} and also \cite{nasi4}. 
We recall that given an observable $A \in \A_V$ and a 
set $Y \subset V$, one may define an observable 
\begin{equation}
\langle A \rangle_Y = \int_{\mathcal{U}(Y^c)} U^* A U \, \mu(dU),
\end{equation}
where $\mathcal{U}(Y^c)$ denotes the group of unitary operators over
the Hilbert space $\mathcal{H}_{Y^c}$ and $\mu$ is the associated,
normalized Haar measure. It is easy to see that for any $A \in
\A_V$, the observable $\langle A \rangle_{Y}$ has been localized to
$Y$ in the sense of supports, i.e.,  $\langle A \rangle_{Y} \in \A_Y$.

Let us define
\begin{equation}
P_{\X_{\rm bd}} \, = \, \langle \hat{\P}_{\alpha}
\rangle_{\X_{\rm bd}(3\ell)}.
\end{equation}
Clearly $P_{\X_{\rm bd}} \in A_{\X_{\rm bd}(3 \ell)}$ and
$\|P_{\X_{\rm bd}}\|\leq 1$ as desired.

The difference between $\hat{\P}_{\alpha}$ and $P_{\X_{\rm
bd}}$ can be expressed in terms of a commutator, in fact,
\begin{equation} \label{eq:opdif}
\hat{\P}_{\alpha} \, - \, P_{\X_{\rm bd}} \, = \,
\int_{\mathcal{U}(\X_{\rm bd}(3\ell)^c)} U^* \left[
\hat{\P}_{\alpha},  U  \right] \, \mu(dU).
\end{equation}
Thus, to show that the difference between $\hat{\P}_{\alpha}$ and
$P_{\X_{\rm bd}}$ is small (in norm), we need only estimate
the commutator of $\hat{\P}_{\alpha}$  with an arbitrary unitary
supported in $\X_{\rm bd}(3\ell)^c$.

This follows from a Lieb-Robinson estimate. Note that
\begin{equation} \label{eq:comm}
\left[ \hat{\P}_{\alpha}, U \right] \, = \, \sqrt{ \frac{ \alpha}{
\pi}} \int_{- \infty}^{\infty} \left[ e^{i(M_{\X} +M_{\X_{\rm
bd}}+M_{\X^c})t} e^{-i(M_{\X}+M_{\X^c})t}, U \right] e^{- \alpha
t^2}  \, dt.
\end{equation}
To estimate the integrand, we define the function
\begin{equation}
f(t) \, = \,  \left[ e^{i(M_{\X} +M_{\X_{\rm bd}}+M_{\X^c})t}
e^{-i(M_{\X} +M_{\X^c})t}, U \right] .
\end{equation}
A short calculation demonstrates that
\begin{equation} \label{eq:fder}
f'(t) \, = \, i   \left[ e^{i(M_{\X} +M_{\X_{\rm bd}}+M_{\X^c})t}
M_{\X_{\rm bd}} e^{-i(M_{\X} +M_{\X^c})t}, U \right] .
\end{equation}
The form of the derivative appearing in (\ref{eq:fder}) suggests
that we define the evolution
\begin{equation}
\alpha_t(A) = e^{i(M_{\X} +M_{\X_{\rm bd}}+M_{\X^c})t}Ae^{-i(M_{\X}
+M_{\X_{\rm bd}}+M_{\X^c})t} \quad \quad \mbox{for any } \quad A \in
\mathcal{A}_V.
\end{equation}

With this in mind, we rewrite
\begin{eqnarray}
f'(t) & = &  i   \left[ e^{i(M_{\X} +M_{\X_{\rm bd}}+M_{\X^c})t} M_{\X_{\rm bd}} e^{-i(M_{\X} +M_{\X^c})t}, U \right] \nonumber \\
& = &  i   \left[ \alpha_t( M_{\X_{\rm bd}}) e^{i(M_{\X} +M_{\X_{\rm bd}}+M_{\X^c})t} e^{-i(M_{\X} +M_{\X^c})t} , U \right] \nonumber \\
& = &  i \alpha_t( M_{\X_{\rm bd}}) f(t)  + i  \left[ \alpha_t(
M_{\X_{\rm bd}}), U \right] e^{i(M_{\X} +M_{\X_{\rm bd}}+M_{\X^c})t}
e^{-i(M_{\X} +M_{\X^c})t} .
\end{eqnarray}

Written as above, the function $f$ can be bounded using
norm-preservation, see \cite{nasi1} and \cite{nasi4} for details. For example, let  $V(t)$ be the unitary evolution
that satisfies the time-dependent differential equation
\begin{equation}
i \frac{d}{dt} V(t) =  V(t) \alpha_t(M_{\X_{\rm bd}})  \quad \quad
\mbox{with } V(0) = \idty.
\end{equation}
Explicitly, one has that
\begin{equation}
V(t) = e^{i(M_{\X} +M_{\X^c})t} e^{-i(M_{\X} +M_{\X_{\rm
bd}}+M_{\X^c})t}.
\end{equation}
Considering now the product
\begin{equation}
g(s) = V(s)  f(s)
\end{equation}
it is easy to see that
\begin{equation}
g'(s) = V'(s) f(s) + V(s) f'(s) = i V(s)  \left[ \alpha_s(
M_{\X_{\rm bd}}), U \right]  V(s)^*.
\end{equation}
Thus
\begin{equation}
 V(t) f(t) = g(t) - g(0) = \int_0^t g'(s) ds,
\end{equation}
and therefore,
\begin{equation}
f(t) \, = \, i  V(t)^* \int_0^t V(s) \left[ \alpha_s( M_{\X_{\rm
bd}}), U \right] V(s)^* \, ds .
\end{equation}
The bound
\begin{equation} \label{eq:fbd1}
\| f(t) \| \leq \int_0^{|t|} \left\|  \left[ \alpha_s( M_{\X_{\rm
bd}}), U \right] \right\| \, ds.
\end{equation}
readily follows.

Unfortunately, the Lieb-Robinson velocity associated to the dynamics
$\alpha_t( \cdot)$ grows with $\ell$. For this reason, we estimate (\ref{eq:fbd1}) by
comparing back to the original dynamics. Consider the interpolating
dynamics
\begin{equation}
h_s(r) \, = \, \alpha_r \left(  \tau_{s-r}(A) \right),
\end{equation}
for any local observable $A$ and $0 \leq r \leq s$. Here $\tau_r(A)
= e^{iH_V r}Ae^{-i H_V r}$. With $s$ fixed, it is easy to calculate
\begin{equation}
h_s'(r) \, = \, i \alpha_{s-r} \left( \left[ (M_{\X} +M_{\X_{\rm
bd}}+M_{\X^c})-H_V, \tau_r(A) \right] \right).
\end{equation}
We conclude then that
\begin{eqnarray} \label{eq:compbd}
\| \alpha_s(M_{\X_{\rm bd}}) - \tau_s(M_{\X_{\rm bd}}) \| & = & \left\| \int_0^s h_s'(r) \, dr \right\| \nonumber \\
& \leq & 2 \, \| M_{\X_{\rm bd}} \| \, \| H_V - (M_{\X} +M_{\X_{\rm bd}}+M_{\X^c}) \| \, s \nonumber \\
& \leq & 2 \left( J C_{\Phi} | \X_{\rm bd}| \right) \, \left( K | \partial \X| \ell^{3/2} e^{- \frac{\ell}{\xi}} \right) \, s \nonumber \\
& \leq & K  | \partial \X|^2 \ell^{5/2} e^{- \frac{\ell}{ \xi}} \,
s,
\end{eqnarray}
where, for the second inequality above, we used Lemma \ref{lem:bd2}
and the bound
\begin{eqnarray}
\| M_{\X_{\rm bd}} \| & \leq & \sum_{x \in \X_{\rm bd}}
\sum_{\stackrel{X \subset V:}{x \in X}}
\sqrt{\frac{ \alpha}{ \pi}} \int_{- \infty}^{\infty} \left\| \tau_t^{\X_{\rm bd}} \left( \Phi(X) \right) \right\| e^{- \alpha t^2} \, dt \nonumber \\
& \leq & J C_{\Phi} | \X_{\rm bd} | \, .
\end{eqnarray}
{F}rom (\ref{eq:fbd1}), it is clear that
\begin{equation}
\| f(t) \|  \leq  \int_0^{|t|} \left\| \left[ \tau_s(M_{\X_{\rm
bd}}), U \right] \right\| \, ds \, + \,   \int_0^{|t|} \left\|
\left[ \alpha_s(M_{\X_{\rm bd}}) \, - \,  \tau_s(M_{\X_{\rm bd}}), U
\right] \right\| \, ds \, .
\end{equation}
The first term above, we bound, using Theorem~\ref{thm:lrb}, as
follows
\begin{eqnarray}
 \int_0^{|t|} \left\| \left[ \tau_s(M_{\X_{\rm bd}}), U \right] \right\| \, ds & \leq & C | \partial \X_{\rm bd}(2 \ell)| \, \| M_{\X_{\rm bd}} \| \, \int_0^{|t|} e^{- \mu \left( d( \X_{\rm bd}(2 \ell),U) - v|s| \right)} \, ds  \nonumber \\
 & \leq & K | \partial \X|^2 \ell^2 e^{- \mu \ell} \frac{1}{\mu v} \left( e^{\mu v|t|} - 1 \right) \,.
\end{eqnarray}
For the second, we apply (\ref{eq:compbd}). Thus, we have found that
\begin{eqnarray}
\| [ \hat{\P}_{\alpha}, U ] \| & \leq & \sqrt{ \frac{ \alpha}{ \pi}} \int_{- \infty}^{\infty} \| f(t) \| e^{- \alpha t^2} dt \nonumber \\
& \leq &  K | \partial \X|^2 \ell^2 e^{- \mu \ell} \frac{1}{\mu v} \sqrt{ \frac{ \alpha}{ \pi}} \int_{- \infty}^{\infty} e^{\mu v|t|} \, e^{- \alpha t^2} \, dt \nonumber \\
& \mbox{ } & \quad + 2 K  | \partial \X|^2 \ell^{5/2} e^{- \frac{\ell}{ \xi}} \sqrt{ \frac{ \alpha}{ \pi}} \int_{- \infty}^{\infty} t^2 \, e^{- \alpha t^2} \, dt \nonumber \\
& \leq & K \| \partial \X|^2 \ell^{7/2} e^{- \frac{ \ell}{\xi}} \, .
\end{eqnarray}
Combining this with \eqref{eq:opdif}, we obtain that
\begin{equation} \label{eq:goal2}
\left\| P_{\X_{\rm bd}}  - \hat{P}_{\alpha} \right\| \, \leq
\, K | \partial \X|^2 \ell^{7/2} e^{- \frac{\ell}{\xi}},
\end{equation}
and the proof of Theorem~\ref{thm:gsapp} is now complete.

%
%
%
%

\Section{A Technical Estimate}\label{sec:test}

In this section, we will prove Lemma \ref{lem:bd2}.
For convenience, we restate the result.

\begin{lem} 
Let $\mathcal{V}$ be a countable, locally finite set equipped with a metric $d$ for which
(\ref{def:kappa}) holds, and let $\Phi$ be an interaction on $\mathcal{V}$ 
which satisfies assumptions {\bf A1} - {\bf A4}. For the finite 
set $V \subset \mathcal{V}$ from {\bf A4} and any set $\X \subset V$
satisfying (\ref{eq:setbd}) for all $\ell >R$, the estimate 
\begin{equation} \label{eq:totbd5}
\left\| H_V \, - \, \left(M_{\X}( \alpha) + M_{\X_{\rm bd}}(\alpha)
+M_{\X^c}( \alpha) \right) \right\| \, \leq \, K | \partial \X|
\ell^{3/2} e^{- \frac{\ell}{\xi}},
\end{equation}
holds along the parametrization $2 \alpha \epsilon \ell = \mu v^2$ when $\mu \geq 2 \mu_0$. The numbers $\xi$ and $\epsilon$
are defined in terms of the gap $\gamma$ and the quantities $\mu$ and $v$ from the 
Lieb-Robinson estimates as
\begin{equation}
0 < \frac{2}{\xi} = (1- \epsilon) \mu = \frac{ \gamma^2}{ \mu^2 v^2 + \gamma^2}  \mu \, .
\end{equation}
Our bounds on the prefactor $K$ depend on various parameters, e.g. $J^2, C_{\Phi}, N_{\Phi}, \mu, v, \gamma, R, \kappa_{\mu/2}$,
but it is independent of $\ell, V$, and $\X$. Along the same
parametrization, the bound
\begin{equation} \label{eq:Monpsi5}
\max \left\{ \left\| M_{\X}( \alpha) \psi_0 \right\|, \left\|
M_{\X_{\rm bd}}( \alpha) \psi_0 \right\|, \left\| M_{\X^c} ( \alpha)
\psi_0 \right\| \right\}  \, \leq \, K | \partial \X| \ell^{3/2}
e^{- \frac{\ell}{\xi}}.
\end{equation}
also holds.
\end{lem}

To prove Lemma~\ref{lem:bd2}, we begin with a few simple propositions.
Let $V \subset \mathcal{V}$ be the finite set described in {\bf A4}, take $\ell >R$, and 
fix $\X \subset V$ satisfying (\ref{eq:setbd}).
As we demonstrates in Section~\ref{sec:setup}, see also Section~\ref{sec:gsat},
the local Hamiltonian corresponding to $\Phi$ in $V$ can be written as
\begin{equation} \label{eq:hamdec2}
H_V  =  H_{\X_{\rm int}}^b + H_{\X_{\rm bd}} + H_{\X_{\rm ext}}^b
\end{equation}
where the sets $\X_{\rm int}$, $\X_{\rm bd}$, and $\X_{\rm ext}$, and the corresponding local Hamiltonians, each depend on
a length scale $\ell >R$. We begin with a basic commutator estimate.

\begin{proposition} \label{prop:commutator}
Let $\Phi$ be an interaction on $\mathcal{V}$ which satisfies assumptions {\bf
A1} - {\bf A4}. Let $V \subset \mathcal{V}$ be the finite set described in {\bf A4}, take $\ell >R$, and 
fix $\X \subset V$ satisfying (\ref{eq:setbd}). One has that
\begin{equation} \label{eq:hamcombd}
\max \left\{ \| [H_V, H_{\X_{\rm int}}^b] \|,  \| [H_V,H_{\X_{\rm bd}}] \|,  \| [H_V,H_{\X_{\rm ext}}^b] \|   \right\} \leq 4 C J^2 C_{\Phi} N_{\Phi} | \partial \X| \ell \, ,
\end{equation}
where the quantities $C_{\Phi}$ and $N_{\Phi}$ are as introduced in \eqref{eq:finiteint1} and \eqref{eq:finiteint2}.
\end{proposition}

\begin{proof}
Define
\begin{equation}
C_{\rm int} = \left[H_V, H_{\X_{\rm int}}^b \right], \quad C_{\rm
bd} = \left[H_V, H_{\X_{\rm bd}} \right], \quad \mbox{and} \quad
C_{\rm ext} = \left[H_V, H_{\X_{\rm ext}}^b \right].
\end{equation}
{F}rom the definitions, (\ref{eq:xint}), (\ref{eq:xbd}), and
(\ref{eq:xext}), and the fact that $\ell >R$, it is clear that
$C_{\rm int} = \left[H_{\X_{\rm bd}}, H_{\X_{\rm int}}^b \right]$,
$C_{\rm ext} = \left[H_{\X_{\rm bd}}, H_{\X_{\rm ext}}^b \right]$,
and $- C_{\rm bd} = C_{\rm int} + C_{\rm ext}$. Clearly,
\begin{equation}
C_{\rm int} = \left[ H_{\X_{\rm bd}}, H_{\X_{\rm int}}^b \right] =
\sum_{\stackrel{X \subset V:}{X \cap \partial \X_{\rm int} \neq
\emptyset}}  \left[ H_{\X_{\rm bd}}, \Phi(X) \right],
\end{equation}
and therefore,
\begin{eqnarray}
\| C_{\rm int} \| & \leq & \sum_{x \in \partial \X_{\rm int}} \sum_{\stackrel{X \subset V:}{x \in X}} \left\| \left[ H_{\X_{\rm bd}}, \Phi(X) \right] \right\| \nonumber \\
& \leq &  \sum_{x \in \partial \X_{\rm int}} \sum_{\stackrel{X \subset V:}{x \in X}}  \sum_{z \in X} \sum_{\stackrel{Y \subset V:}{z \in Y}} \left\| \left[ \Phi(Y), \Phi(X) \right] \right\| \nonumber \\
& \leq & 2 J^2  \sum_{x \in \partial \X_{\rm int}} \sum_{\stackrel{X \subset V:}{x \in X}} \sum_{z \in X} \sum_{\stackrel{Y \subset V:}{z \in Y}} \chi_{\Phi}(Y) \chi_{\Phi}(X)  \nonumber \\
& \leq & 2 J^2 C_{\Phi} N_{\Phi} | \partial \X_{\rm int} |.
\end{eqnarray}
A similar estimates applies to $C_{\rm ext}$. Using (\ref{eq:setbd}), the proof is complete.
\end{proof}

As is discussed in Section~\ref{sec:setup}, the first step in the proof of 
Theorem~\ref{thm:gsapp} is to introduce the smearing operation 
$( \cdot )_{\alpha}$, see (\ref{eq:smear}). One consequence of this definition is
\begin{equation} \label{eq:hamalpha}
H_V  =  (H_{\X_{\rm int}}^b)_{\alpha} + (H_{\X_{\rm bd}})_{\alpha} +
(H_{\X_{\rm ext}}^b)_{\alpha}.
\end{equation}
When applied to the ground state, each of the terms above can be estimated, in norm, in terms of the
system's gap. This is the content of Proposition~\ref{prop:bd} below.

\begin{proposition} \label{prop:bd}
Let $\Phi$ be an interaction on $\mathcal{V}$ which satisfies assumptions {\bf A1} - {\bf A4}. For any finite set $\X \subset V \subset \mathcal{V}$, $\alpha >0$, and  $A\in\{
H_{\X_{\rm int}}^b, H_{\X_{\rm bd}}, H_{\X_{\rm ext}}^b \}$,
\begin{equation}
\| ( A)_{\alpha} \psi_0 \| \, \leq \, \frac{ \| [H_V, A ] \| }{
\gamma} e^{- \frac{\gamma^2}{4 \alpha}} .
\end{equation}
\end{proposition}

\begin{proof}
Let  $A\in\{ H_{\X_{\rm int}}^b, H_{\X_{\rm bd}}, H_{\X_{\rm ext}}^b
\}$. As previously discussed, see (\ref{eq:0gse}), $0 = \langle
\psi_0, A \psi_0 \rangle = \langle \psi_0, (A)_{\alpha} \psi_0
\rangle$. Setting $C_A = [H_V, A]$, we find that
\begin{eqnarray}
\| ( A )_{\alpha} \psi_0 \| & \leq & \frac{1}{\gamma} \| H_V ( A)_{\alpha} \psi_0 \| \nonumber \\
& = & \frac{1}{\gamma} \| [ H_V, (A)_{\alpha}] \psi_0 \| \nonumber \\
& = & \frac{1}{\gamma} \| (C_A)_{\alpha} \psi_0 \|.
\end{eqnarray}

The value of $\| (C_A)_{\alpha} \psi_0 \|$ can be estimated using
the spectral theorem. For any vector $f$, one has that
\begin{eqnarray}
\langle f, (C_A)_{\alpha} \psi_0 \rangle & = & \sqrt{ \frac{\alpha}{\pi}} \int_{- \infty}^{\infty} e^{- \alpha t^2} \langle f, \tau_t^V(C_A) \psi_0 \rangle dt \nonumber \\
& = &  \sqrt{ \frac{\alpha}{\pi}} \int_{- \infty}^{\infty} e^{- \alpha t^2} \langle f, e^{it H_V}C_A \psi_0 \rangle dt \nonumber \\
& = &  \sqrt{ \frac{\alpha}{\pi}} \int_{- \infty}^{\infty} e^{- \alpha t^2} \int_{\gamma}^{\infty} e^{it \lambda} d \langle f, E_{\lambda} C_A \psi_0 \rangle dt \nonumber \\
& = & \int_{\gamma}^{\infty} \sqrt{ \frac{ \alpha}{ \pi}} \int_{- \infty}^{\infty} e^{- \alpha t^2} e^{it \lambda} \, dt \, d \langle f, E_{\lambda} C_A \psi_0 \rangle \nonumber \\
& = & \int_{\gamma}^{\infty} e^{- \frac{ \lambda^2}{4 \alpha}} \, d
\langle f, E_{\lambda} C_A \psi_0 \rangle.
\end{eqnarray}
In the third equality above we have introduced the notation
$E_{\lambda}$ for the spectral projection corresponding to the
self-adjoint operator $H_V$, and we also used that $\langle \psi_0,
C_A \psi_0 \rangle = 0$. The last equality is a basic result
concerning Fourier transforms of gaussians (and re-scaling), (see
e.g.(5.59) in \cite{nasi3}). We conclude then that
\begin{equation}
\left| \langle f, (C_A)_{\alpha} \psi_0 \rangle \right| \, \leq \,
e^{ - \frac{ \gamma^2}{4 \alpha}} \| f \| \, \| C_A \|,
\end{equation}
and with $f = (C_A)_{\alpha} \psi_0$ we find that
\begin{equation}
\| (C_A)_{\alpha} \psi_0 \| \, \leq \, e^{- \frac{ \gamma^2}{4
\alpha}} \| C_A \|.
\end{equation}
Putting everything together,  we have shown that
\begin{equation} \label{eq:thlbd}
\| ( A)_{\alpha} \psi_0 \| \leq \frac{1}{ \gamma} \| (C_A)_{\alpha}
\psi_0 \| \leq \frac{e^{- \frac{ \gamma^2}{4 \alpha}}}{\gamma} \|
C_A \|.
\end{equation}
\end{proof}

The next step in the proof of Theorem~\ref{thm:gsapp} is to make a 
strictly local approximation of the smeared terms appearing in (\ref{eq:hamalpha}) above.
We defined these local approximations in Section~\ref{sec:setup}, see (\ref{eq:locham1}), and 
Proposition~\ref{prop:bd2} below provides an explicit estimate.

\begin{proposition} \label{prop:bd2} 
Let $\mathcal{V}$ be a countable, locally finite set equipped with a metric $d$ for which
(\ref{def:kappa}) holds, and let $\Phi$ be an interaction on $\mathcal{V}$ which satisfies assumptions {\bf
A1} - {\bf A4}. Let $V \subset \mathcal{V}$ be the finite set described in {\bf A4}, take $\ell >R$, and 
fix $\X \subset V$ satisfying (\ref{eq:setbd}). The estimate 
\begin{equation}
\begin{split}
\max \left\{ \left\| ( H_{\X_{\rm int}}^b)_{\alpha} \, - \, M_{\X}(\alpha) \right\|,  
\left\| ( H_{\X_{\rm bd}( \ell) })_{\alpha} \, - \, M_{\X_{\rm bd}(2 \ell)}(\alpha) \right\|,  
\left\| ( H_{\X_{\rm ext}}^b)_{\alpha} \, - \, M_{\X^c}(\alpha) \right\| \right\}  \\
\leq 2 C \frac{J^2 C_{\Phi} N_{\Phi}}{ \mu v} | \partial \X|
\left( 4 \sqrt{ \frac{2 \epsilon \mu}{\pi}} \, \ell^{3/2} + c \kappa_{\mu/2} e^{3 \mu R} \right) \, e^{- \frac{\ell}{\xi}}. \quad \quad \quad
\end{split}
\end{equation}
holds along the parametrization $2 \alpha \epsilon \ell = \mu v^2$ where $\mu \geq 2 \mu_0$. Here
\begin{equation}
0 < \frac{2}{\xi} = (1- \epsilon) \mu = \frac{ \gamma^2}{ \mu^2 v^2 + \gamma^2} \mu \, .
\end{equation} 
\end{proposition}

\begin{proof}
We will prove the estimate for $(H_{\X_{\rm int}}^b)_{\alpha} -
M_{\X}(\alpha)$, the other bounds follow similarly. Note that
\begin{equation}
\left\| ( H_{\X_{\rm int}}^b)_{\alpha} \, - \, M_{\X}(\alpha)
\right\| \, \leq \, \sqrt{ \frac{ \alpha}{ \pi}} \int_{-
\infty}^{\infty} \left\| \tau_t( H_{\X_{\rm int}}^b) \, - \,
\tau_t^{\X}(H_{\X_{\rm int}}^b) \right\| \, e^{- \alpha t^2} dt.
\end{equation}
To bound the integral above, we introduce a parameter $T>0$. For
$|t| > T$, we use that
\begin{eqnarray}
\left\| \tau_t \left( H_{\X_{\rm int}}^b   \right)  \, - \,
\tau_t^{\X} \left( H_{\X_{\rm int}}^b \right) \right\| \, &\le&
\int_0^{|t|}
\left\| \frac{d}{ds}\Big( \tau_{s} (H_{\X_{\rm int}}^b)  \, - \, \tau_{s}^{\X}(H_{\X_{\rm int}}^b) \Big)\right\| ds \nonumber \\
& \leq & 2 \|[H_V, H_{\X_{\rm int}}^b]\| |t|.
\end{eqnarray}
{F}rom this, it follows readily that
\begin{equation}
 \sqrt{ \frac{ \alpha}{ \pi}} \int_{|t|>T} \left\| \tau_t ( H_{\X_{\rm int}}^b)  \, - \, \tau_t^{\X}( H_{\X_{\rm int}}^b) \right\|
\, e^{- \alpha t^2} dt  \leq  \frac{8 C J^2 C_{\Phi} N_{\Phi} |
\partial \X|}{\sqrt{\alpha\pi}} \ell e^{- \alpha T^2}, 
\end{equation}
where we used Proposition~\ref{prop:commutator}.

For $|t| \leq T$, the estimate below is an immediate consequence of
Lemma 3.1 in \cite{nasi4}:
\begin{equation} \label{eq:smt}
\left\| \tau_t ( H_{\X_{\rm int}}^b )  \, - \, \tau_t^{\X}(
H_{\X_{\rm int}}^b) \right\| \, \leq \, \int_0^{|t|} \left\| \left[
H_V-H_{\X}, \tau_s^{\X} \left( H_{\X_{\rm int}}^b \right)
\right]\right\| \, ds.
\end{equation}
The above commutator may be written as
\begin{equation}  \label{eq:com2}
\left[ H_V - H_{\X}, \tau_s^{\X} \left( H_{\X_{\rm int}}^b \right)
\right] \, = \,
 \sum_{ \stackrel{X \subset V:}{X \cap \X \neq \emptyset, X \cap V \setminus \X \neq \emptyset}} \sum_{\stackrel{Y \subset V:}{Y \cap \X_{\rm int} \neq \emptyset}} \left[ \Phi(X), \tau_s^{\X} \left( \Phi(Y) \right)
 \right] \, .
\end{equation}
Estimating this, we find that
\begin{eqnarray}
\left\| \left[ H_V - H_{\X}, \tau_s^{\X} \left( H_{\X_{\rm int}}^b
\right) \right]  \right\| & \leq &  \sum_{x \in \partial \X}
\sum_{y \in \X_{\rm int}} \sum_{\stackrel{X \subset V:}{x \in X}} \sum_{\stackrel{Y \subset V:}{y \in Y}}  \left\|  \left[ \Phi(X), \tau_s^{\X} \left( \Phi(Y) \right) \right] \right\| \nonumber \\
 & \leq & c J^2  \sum_{x \in \partial \X}
\sum_{y \in \X_{\rm int}} \sum_{\stackrel{X \subset V:}{x \in X}}
\sum_{\stackrel{Y \subset V:}{y \in Y}} \chi_{\Phi}(X)
\chi_{\Phi}(Y) \min \left[ | \partial_{\Phi} X|, | \partial_{\Phi} Y| \right] e^{-
\mu \left( d(X,Y) - v |s| \right)} \, ,
\end{eqnarray}
where we have used the Lieb-Robinson bound, i.e.
Theorem~\ref{thm:lrb} with $\mu \geq \mu_0$ from (\ref{def:kappa}). Since $\Phi$ has a finite range $R$, it is
clear that $d(x,y) - 2 R \leq d(X, Y)$. This implies that
\begin{equation}
\left\| \left[ H_V - H_{\X}, \tau_s^{\X} \left( H_{\X_{\rm int}}^b
\right) \right]  \right\| \leq  c J^2 C_{\Phi} N_{\Phi} e^{2 \mu R}
\sum_{x \in \partial \X} \sum_{y \in \X_{\rm int}} e^{- \mu \left(
d(x,y) - v |s| \right)} \, .
\end{equation}
Now, for each fixed $x \in \partial \X$, $d(x,y) \geq \ell - 2 R$.
Summing on all $y$ yields
\begin{equation}
\left\| \left[ H_V - H_{\X}, \tau_s^{\X} \left( H_{\X_{\rm int}}^b
\right) \right]  \right\| \leq  c \kappa_{\mu/2} J^2 C_{\Phi} N_{\Phi} e^{3 \mu R}
| \partial \X| e^{- \mu \left( \ell/2 - v |s| \right)} \, ,
\end{equation}
for $\mu \geq 2 \mu_0$. Comparing back to (\ref{eq:smt}), this bound demonstrates that
\begin{equation} \label{eq:smt2}
\left\| \tau_t ( H_{\X_{\rm int}}^b )  \, - \, \tau_t^{\X}(
H_{\X_{\rm int}}^b) \right\| \, \leq \, c \kappa_{\mu/2} J^2 C_{\Phi} N_{\Phi} e^{3
\mu R} | \partial \X| e^{- \mu \ell /2 } \frac{(e^{\mu v |t|} - 1)}{ \mu
v}  \, ,
\end{equation}
and hence, for any $T>0$, we have that
\begin{eqnarray}
\sqrt{ \frac{ \alpha}{ \pi}} \int_{-T}^{T} \left\| \tau_t (
H_{\X_{\rm int}}^b )  \, - \, \tau_t^{\X}( H_{\X_{\rm int}}^b)
\right\|  \, e^{- \alpha t^2} dt & \leq &  c \kappa_{\mu/2} e^{3 \mu R} \frac{J^2 C_{\Phi} N_{\Phi}}{ \mu v}
 | \partial \X| e^{- \mu \ell /2 } \, \sqrt{ \frac{ \alpha}{ \pi}} \int_{-T}^{T} e^{ \mu v |t|} e^{- \alpha t^2} \, dt \nonumber \\
& \leq &   2  c \kappa_{\mu/2} e^{3 \mu R} \frac{J^2 C_{\Phi} N_{\Phi}}{ \mu v}
| \partial \X| e^{- \mu \ell /2} \, e^{\frac{\mu^2 v^2}{4 \alpha}} \, .
\end{eqnarray}

Adding our results for large and small $T$, it is clear that
\begin{equation} \label{eq:hlmldif}
\begin{split}
\left\| ( H_{\X_{\rm int}}^b)_{\alpha} \,  - \, M_{\X}(\alpha) \right\| \, & \leq \, \\
 2 \frac{J^2 C_{\Phi} N_{\Phi}}{ \mu v} | \partial \X| e^{- \alpha T^2} & \left( \frac{4 C \mu v \ell}{\sqrt{\alpha\pi}} +  c \kappa_{\mu/2} e^{3 \mu R} \cdot
\mbox{exp} \left[ \alpha T^2 - \frac{\mu}{2} \left( \ell -\frac{\mu v^2}{ 2
\alpha} \right) \right] \right).
\end{split}
\end{equation}
We now choose a parametrization which illustrates the decay. Let $\alpha$ satisfy 
$$
\frac{\mu v^2}{2 \alpha} =  \epsilon \ell  \quad \mbox{ with } \quad 0 < \epsilon = \left(1+ \frac{ \gamma^2}{ \mu^2 v^2} \right)^{-1} < 1,
$$
and choose $T$ to satisfy the equation
\begin{equation}
\alpha T^2 - \frac{\mu}{2} \left( \ell - \frac{ \mu v^2}{ 2 \alpha} \right) =
0.
\end{equation}
In this case, $2 \alpha T^2 = (1 - \epsilon) \mu \ell$, and thus 
\begin{equation}
\left\| ( H_{\X_{\rm int}}^b)_{\alpha} \,  - \, M_{\X}(\alpha)
\right\| \leq 2 \frac{J^2 C_{\Phi} N_{\Phi}}{ \mu v} | \partial \X| e^{- (1- \epsilon) \mu \ell /2 }
\left( 4 C \sqrt{ \frac{2 \epsilon \mu}{\pi}} \, \ell^{3/2} + c \kappa_{\mu/2} e^{3 \mu R} \right) .
\end{equation}
\end{proof}

Equipped with theses estimates, we are now ready to prove Lemma
\ref{lem:bd2}

\begin{proof}(of Lemma~\ref{lem:bd2}:)
Using (\ref{eq:hamalpha}) and Proposition~\ref{prop:bd2}, we find
that
\begin{equation} \label{eq:hamappbd}
\begin{split}
\left\| H_V \, - \, \left(M_{\X}( \alpha) + M_{\X_{\rm bd}}(\alpha)
+M_{\X^c}( \alpha)
\right) \right\| \,  \hspace{3.5cm} \\
\leq \, 6 C \frac{J^2 C_{\Phi} N_{\Phi}}{ \mu v} | \partial \X|
\left( 4 \sqrt{ \frac{2 \epsilon \mu}{\pi}} \, \ell^{3/2} + c \kappa_{\mu/2} e^{3 \mu R} \right) \, e^{- \frac{\ell}{\xi}} \, ,
\end{split}
\end{equation}
along the parametrization $2 \alpha \epsilon \ell = \mu v^2$. This is the bound claimed in (\ref{eq:totbd5}).

To see that (\ref{eq:Monpsi5}) is true, note that, for example,
\begin{equation} \label{eq:mxpsi}
\| M_{\X}( \alpha) \psi_0 \| \, \leq \, \| (H^b_{\X_{\rm
int}})_{\alpha} \psi_0 \| \, + \, \| \left( ( H^b_{\X_{\rm
int}})_{\alpha} - M_{\X}(\alpha) \right) \psi_0 \| \, .
\end{equation}
Using Propositions~\ref{prop:commutator} and Proposition~\ref{prop:bd}, it is clear
that
\begin{eqnarray}
 \| (H^b_{\X_{\rm int}})_{\alpha} \psi_0 \| & \leq & \frac{  \left\| \left[ H_V, H^b_{\X_{\rm int}} \right] \right\|}{ \gamma} e^{- \frac{ \gamma^2}{4 \alpha}} \nonumber \\
 & \leq & 4 C \frac{J^2 C_{\Phi} N_{\Phi} }{ \gamma} | \partial \X | \ell e^{- \frac{\ell}{ \xi}} \,,
\end{eqnarray}
along the parametrization $2 \alpha \epsilon \ell = \mu v^2$. Combining this with Proposition~\ref{prop:bd2}, we
have that
\begin{equation}
\| M_{\X}( \alpha) \psi_0 \| \, \leq \, 2 C J^2 C_{\Phi} N_{\Phi} | \partial \X | \,
\left( \frac{ 2 \ell}{ \gamma}  + \frac{4}{ \mu v} \sqrt{ \frac{ 2 \epsilon \mu}{ \pi}} \ell^{3/2} +  \frac{c \kappa_{\mu/2} e^{3 \mu R}}{ \mu v} \right) \, e^{- \frac{\ell}{ \xi}} \, .
\end{equation}
A similar bound applies to both $\| M_{\X_{\rm bd}}( \alpha) \psi_0
\|$ and $\| M_{\X^c}( \alpha) \psi_0 \|$. This completes the proof of the
Lemma~\ref{lem:bd2}.
\end{proof}

\noindent {\bf Acknowledgements:}

E. H would like to acknowledge support through a Junior Research
Fellowship at the Erwin Schrodinger Institute in Vienna, where part of this work
was done. A part of this work was also supported by the National Science
Foundation: B.N. and S. M. under Grants \#DMS-0605342 and \#DMS-0757581
and R.S. under Grant \#DMS-0757424.

\end{document}